\documentclass{amsart}
\usepackage{amssymb}
\usepackage{lmodern}
\usepackage[T1]{fontenc}
\usepackage[utf8]{inputenc}
\usepackage{microtype}
\usepackage[usenames, dvipsnames]{xcolor}
\usepackage[backref=page, colorlinks=true]{hyperref}
\hypersetup{
 linkcolor={red!50!black},
 citecolor={green!50!black},
 urlcolor={blue!80!black}
}
\usepackage{tikz-cd}
\usepackage[margin=1.3in]{geometry}
\usepackage[capitalize, noabbrev]{cleveref} 

\newcommand{\beq}{\begin{eqnarray}}
	\newcommand{\eeq}{\end{eqnarray}}

\newcommand{\bsp}{\begin{aligned}}
	\newcommand{\esp}{\end{aligned}}

\usepackage{xcolor}
\definecolor{darkblue}{rgb}{0.,0.,0.4}
\definecolor{darkred}{rgb}{0.5,0.,0.}
\definecolor{BlueViolet}{RGB}{138,43,226}
\definecolor{SkyBlue}{RGB}{30,144,255}
\definecolor{DarkGreen}{RGB}{0,100,0}

\usepackage{mathrsfs}
\usepackage{subcaption}
\usepackage{comment}

\newcommand{\Z}{\mathbb{Z}}

\usepackage{pst-all}
\usepackage{leftidx}
\usepackage[off]{auto-pst-pdf} 
\usepackage{yfonts}
\makeatletter

\usepackage{subcaption}

\usepackage{svg}

\usepackage[capitalize, noabbrev]{cleveref} 

\numberwithin{equation}{section}
\usepackage{amsthm}
\newtheorem{thm}[equation]{Theorem}

\newtheorem{lem}[equation]{Lemma}

\newtheorem{prop}[equation]{Proposition}
\newtheorem{cor}[equation]{Corollary}

\newtheorem*{thm*}{Theorem}

\newcounter{mainthm}

\newtheorem{maintheorem}[mainthm]{Theorem}

\newtheorem*{slogan}{Slogan}

\theoremstyle{definition}
\newtheorem{defn}[equation]{Definition}
\newtheorem{example}[equation]{Example}
\theoremstyle{remark}
\newtheorem{rem}[equation]{Remark}

\crefname{cor}{Corollary}{Corollaries}

\crefname{thm}{Theorem}{Theorems}
\crefname{prop}{Proposition}{Propositions}
\crefname{lem}{Lemma}{Lemmas}

\crefname{defn}{Definition}{Definitions}
\usepackage{mathtools}

\usepackage{tikz-cd}
\usepackage{subcaption}
\usepackage{setspace}
\usepackage{booktabs}
\usetikzlibrary{calc}

\usepackage[mathscr]{eucal}
\makeatletter
\def\instring#1#2{TT\fi\begingroup
  \edef\x{\endgroup\noexpand\in@{#1}{#2}}\x\ifin@}
\def\isuppercase#1{%
  \instring{#1}{ABCDEFGHIJKLMNOPQRSTUVWXYZ}%
}%
\newcommand{\C@lIfUpper}[1]{
 \if\isuppercase{#1}\mathscr{#1}%
 \else #1%
 \fi
}
\newcommand{\cat}[1]{\mathit{\@tfor\next:=#1\do{\C@lIfUpper{\next}}}}
\makeatother

\usepackage{xparse}
\newcommand{\bSigma}{\mathbf{\Sigma}}

\newcommand{\cZ}{\mathcal Z}
\newcommand{\rB}{\mathrm B}

\newcommand{\Spin}{\mathrm{Spin}}

\newcommand{\pt}{\mathrm{pt}}

\newcommand{\matt}[1]{{\bf \color{orange} [MY: #1]}}
\newcommand{\dev}[1]{{\bf \color{blue} [DS: #1]}}

\newcommand{\id}{\mathrm{id}}

\newcommand{\Vect}{\mathbf{Vect}}

\newcommand{\nVect}{\mathbf{nVect}}

\newcommand{\Spaces}{\Cat_{(\infty,0)}}

\newcommand{\fib}{\mathbf{Fib}}
\newcommand{\cob}{\mathrm{Cob}}

\newcommand{\cC}{\mathcal C}
\newcommand{\cD}{\mathcal D}

\newcommand{\cX}{\mathcal X}

\newcommand{\Hom}{\mathrm{Hom}}

\newcommand{\sAut}{\mathscr{A}ut}

\newcommand{\Fun}{\mathrm{Fun}}
\newcommand{\End}{\mathrm{End}}

\newcommand\Set{\mathbf{Set}}
\newcommand\Cat{\mathbf{Cat}}

\newcommand{\Th}{\mathbf{Th}}

\newcommand{\un}{\mathbf{Un}}
\newcommand{\st}{\mathbf{Str}}
\newcommand{\RMod}{\mathbf{RMod}}
\newcommand{\lin}{\mathrm{lin}}
\newcommand{\oblv}{\mathrm{oblv}}
\newcommand{\rint}{\mathrm{int}}
\newcommand{\nnVect}{\mathbf{(n+1)Vect}}
\newcommand{\bTheta}{\mathbf{\Theta}}

\newcommand{\alignedintertext}[1]{%
  \noalign{%
    \vskip\belowdisplayshortskip
    \vtop{\hsize=0.75\linewidth#1\par
    \expandafter}%
    \expandafter\prevdepth\the\prevdepth
  }%
}

\newcommand\MAILTO[1]{\href{mailto:#1}{\nolinkurl{#1}}}

\DeclareDocumentCommand{\shortexact}{s O{} O{} mmmm}{
\IfBooleanTF{#1}{ 
\begin{tikzcd}[ampersand replacement=\&]
	{1} \& {#4} \& {#5} \& {#6} \& {1#7}
	\arrow[from=1-1, to=1-2]
	\arrow["#2", from=1-2, to=1-3]
	\arrow["#3", from=1-3, to=1-4]
	\arrow[from=1-4, to=1-5]
\end{tikzcd}
}{ 
\begin{tikzcd}[ampersand replacement=\&]
	{0} \& {#4} \& {#5} \& {#6} \& {0#7}
	\arrow[from=1-1, to=1-2]
	\arrow["#2", from=1-2, to=1-3]
	\arrow["#3", from=1-3, to=1-4]
	\arrow[from=1-4, to=1-5]
\end{tikzcd}
}}

\usepackage{xspace}

\newcommand{\Aut}{\mathrm{Aut}}

\begin{document}

\title{A Generalized Crystalline Equivalence Principle}

\author{Devon Stockall}
\address{Centre for Quantum Mathematics,
University of Southern Denmark, Campusvej 55,  5230 Odense M,
Denmark }
\email{stockall@imada.sdu.dk}

\author{Matthew Yu}
\address{Mathematical Institute, University of Oxford,  Woodstock Road, Oxford, UK}
\email{yumatthew70@gmail.com}

\begin{abstract}
We prove a general version of the crystalline equivalence principle which gives an equivalence of categories between a category of TQFTs defined on a generic space with $G$-symmetry, and a category of TQFTs with internal symmetry. We give a definition and classification of anomalies associated to TQFTs in the presence of spatial symmetry, which we then generalize to a definition of an anomaly for a categorical symmetry. 
\end{abstract}

\maketitle

\section{Introduction}\label{section:intro}
The crystalline equivalence principle (CEP) was introduced by Thorngren-Else in \cite{TE1} and further developed in \cite{TE2}.  It proposes a correspondence between \emph{crystalline topological phases}, which are families of topological quantum field theories (TQFTs) parameterized by a (contractable) space with $G$-symmetry, and TQFTs with internal $G$-symmetry.  This is particularly powerful in the study of spatial symmetry-protected topological (SPT) phases \cite{Jiang:2016yea,Cheng:2018aaz,Freed:2019jzd,Zhang:2019zcj,zhang2022real}, using the techniques of homotopy theory explained in \cite{FH:2016rqq}, and may also provide insight into SymTFTs for spatial symmetries, as studied in \cite{Pace:2024acq,Pace:2025hpb}.  The fermionic analogue has been applied to classify fermionic SPTs with spatial symmetries in  \cite{Debray:2021rik}.  

More generally, one expects a correspondence between families of TQFTs parameterized by a space $\mathcal{X}$ with symmetry group $G$, and families of TQFTs parameterized by the \emph{homotopy quotient} $\mathcal{X}/\!/G$ (\cite[Theorem 3]{TE1}).  We will call this correspondence the \textit{generalized crystalline equivalence principle} (GCEP).\footnote{This name is motivated by condensed matter physics, where the space $\cX$ is treated as a lattice in which the topological phase is defined on. In more abstract formulations of quantum spin systems, the lattice is taken to be any metric space. For our results, we can be agnostic about the nature of $\cX$.}  The goal of this work is to elucidate the physical content of the GCEP by developing a precise mathematical framework in which it can be rigorously formulated and proven.  Within this framework, we also investigate \emph{anomalous} families of theories, and demonstrate how an anomalous family may be naturally understood as a relative TQFT.  

\subsection{Main Results}\label{subsection:results}

Let $\bTheta$ be a target for $n$-dimensional framed (resp. bosonic/fermionic) TQFTs, i.e. a symmetric monoidal $(\infty,n)$-category with duals (Definition \ref{def:withduals}) (resp. $\mathrm{SO}(n)/\mathrm{Spin}(n)$ fixed points of such a category).  Let $\cX$ be a space.  We define a notion of \emph{$\cX$-family of TQFTs valued in $\bTheta$} (\Cref{def:Xfamily}).  If $\cX$ is equipped with an action of group $G$, one can define the \emph{homotopy quotient} $\cX/\!/G$, and specify what it means for a $\cX$-family of TQFTs $G$-invariant (\Cref{def:GinvntXfamily}).  After giving precise mathematical meaning to each of these physical concepts, the following restatement of \cite[Theorem 3]{TE1} is an immediate consequence of the universal property of (homotopy) colimits. 

\begin{maintheorem}\label{thm:GCEP}
 Let $\cX$ be a space equipped with $G$-action.  There is an equivalence of $(\infty,n)$-categories between the category of $G$-invariant $\cX$-families of TQFTs valued in $\bTheta$, and the category of $\cX/\!/G$-families of TQFTs valued in $\bTheta$. 
\end{maintheorem} 

\begin{proof}
We prove this fact more generally for $\cX$ being an $(\infty,n)$-category equipped with $G$-action.  Denote by $\underline{\bTheta}$ the constant functor $\rB G\to \Cat_{(\infty,n)}$, which takes the morphisms $g\in \Hom_{\rB G}(\pt,\pt)$ to $\id_\bTheta$.  By the universal property of homotopy colimit, we have an equivalence of $(\infty,n)$-categories 
\begin{equation}
\Hom_{\Fun(\rB G,\Cat_{(\infty,n)})}(\cX,\underline{\bTheta})\simeq \Hom_{\Cat_{(\infty,n)}}(\mathrm{colim}_{\rB G} \cX,\bTheta)=:\Hom_{\Cat_{(\infty,n)}}(\cX/\!/G, \bTheta). 
\end{equation}
According to \Cref{def:Xfamily}, \Cref{def:GinvntXfamily}, and \Cref{rem:cobordismreduction}, this completes the proof.
\end{proof}

One should note that the above theorem is an equivalence of categories, and the TQFTs themselves on either side of the equivalence are \textbf{not} identified.  In the case that the space $\cX$ is contractable, we recover the statement of the Crystalline Equivalence Principle, as introduced in \cite{TE1}, and also proven in \cite[Section 7.3]{Kong:2021ups} for anti-unitary symmetries using different methods. 

\begin{cor}\label{cor:CEP}
   Suppose that $\cX$ is a contractable space equipped with $G$-action.  There is an equivalence of $(\infty,n)$-categories between the category of $G$-invariant $\cX$-families of TQFTs valued in $\bTheta$, and TQFTs valued in $\bTheta$, equipped with internal $G$-symmetry. 
\end{cor}

\begin{proof}
   In the case that $\cX$ is contractible, we have an equivalence $\rB G\simeq \cX/\!/G$.  The result follows from  Theorem \ref{thm:GCEP}, after noting that the data of a theory with internal $G$-symmetry (that is, a theory on manifolds with $G$-structure) is equivalent to that of a family of theories over $\rB G$ \cite[Theorem 2.4.26]{LurieTFT}.\footnote{The idea that `the data of a theory with internal $\cX$-symmetry is equivalent to that of a family of theories parameterized by $\cX$' is captured more generally by \cite[Theorem 2.4.18]{LurieTFT}, and has also been expressed and applied in the context of  factorization algebras in \cite{costello2023factorization}.}
\end{proof}
\noindent More generally, it is sensible to consider $G$-\emph{equivariant} families of TQFTs, rather than invariant families.  Given a map $G\to \mathrm{O}(n)$, the canonical action of $\mathrm{O}(n)$ on $\bTheta$ induces an action of $G$ on $\bTheta$.  In this context, a $G$-equivariant analogue of  \Cref{cor:CEP} follows from \cite[Theorem 2.4.18]{LurieTFT}.  As a special case, \emph{bosonic} theories arise by applying this procedure to $\mathrm{SO}(n)\to \mathrm{O}(n)$, and \emph{fermionic} theories arise by applying this procedure to $\mathrm{Spin}(n)\to \mathrm{O}(n)$. 

One can then consider equivariant bosonic of fermionic versions of the Crystalline Equivalence Principle, which arise by considering groups $G$, equipped with a map $G\to \mathrm{SO}(n)$ or $G\to \mathrm{Spin}(n)$, respectively.  This data in the fermionic case is equivalently encoded by $G\to \mathrm{O}(n)$, together with classes $s_1\in \mathrm H^1(\rB G,\mathbb{Z}/2)$ and $\omega_2\in \mathrm H^2( \rB G, \mathbb{Z}/2)$, which encapsulates the data of the two twists.  We will call this data a supergroup, and denote it as $G^f = (G,s_1,\omega_2)$.  This consistent with the data of a spatial fermionic symmetry, as studied in \cite[Section IIIA]{Barkeshli2023}.  The $G^f$ symmetry can mix with the spin structure to give a \textit{twisted spin structure}, where instead of $w_1(TM)=0$ and $w_2(TM)=0$  one allows $w_1(TM)=s_1$ and $w_2(TM)=\omega_2$. See \cite{Decoppet:2024htz} for a more detailed account on supergroups and superspaces. 

A fermionic crystalline topological phase is a fermionic TQFT defined on a space $\cX$ with a $G^f$-symmetry.
Examples of the fermionic CEP state that fermionic TQFTs on a manifold with spatial $G^f$-symmetry is equivalent to  fermionic TQFTs with $\tilde{G}^f$ internal symmetry, where $\tilde{G}^f$ is in general a different supergroup from $G^f$ \cite{Debray:2021rik,Cheng:2024awi,Cheng:2018aaz,Zhang:2019zcj}.

\begin{rem}
 On the side of the target category for TQFTs, one may employ categories enriched over super vector spaces, as in \cite{DY2025}, to accommodate fermionic theories. In our approach, within the framework of the universal target $\mathbf{\Theta}$, one can take into account for the fermionic nature of the TQFT by taking homotopy fixed points with respect to the group $\Spin$. For example, one could consider $\mathbf{\widehat{\Theta}} = \mathbf{nSVect}$, the $n$-category of finite dimensional super $n$-vector spaces over $\mathbb C$. The specific case of $n=4$ has applications to (3+1)d topological orders. The generalized cohomology theory corresponding to $\mathbf{4SVect}^\times$ was denoted $\mathrm{SW}$  in \cite{JF:2020twl,JF:2021tbq,Decoppet:2022dnz,DY2025,Debray:2025kfg}. In \cite{Stockall:2025ngz} the authors also used geometric categories for the target category, in order to accommodate theories with continuous symmetries. In general the target will change depending on what structures and properties of TQFTs one wants to capture.
\end{rem}

After relating equivariant $\cX$-families of TQFTs to $\cX/\!/G$-families of TQFTs, we present an independent approach for studying 't Hooft anomalies for families of TQFTs parameterized by an $\infty$-groupoid.  't Hooft anomalies capture the projectivity of a theory under the action of a group $G$, and are traditionally studied by examining the failure of the partition function to remain invariant under $G$-gauge transformations\footnote{Anomalies can arise even if there is no group acting, see \cite{Freed:2023snr}.}.  To specify anomalies for $\cX$-families of theories, it will be necessary to take another approach.  The subsequent definitions and theorems can be used to classify anomalies for spatial symmetries.

\begin{maintheorem}[Theorem \ref{thm:catofanomalies}]\label{mainthm:classifyanom}
    The category of anomalies for $\cX$-families of TQFTs with value $\bTheta$ is equivalent to the full subcategory of functors 
    $$\Fun(\cX,\rB\sAut(\mathbf{\Theta}))\subset \Fun(\cX,\Spaces)\,,$$
    on those $\alpha \in \Fun(\cX,\Spaces)$ such that $\alpha(y)\simeq \mathbf{\Theta}$ for all $x\in \cX$.  
\end{maintheorem}

Since the above results make sense for anomalies of TQFTs with arbitrary target, in section \S\ref{section:anomaliesSPTs},  we use \ref{thm:catofanomalies} to prove a general theorem about how an anomalous theory can be seen as a relative TQFT.  This is closed related to the theory of anomalies presented in \cite[Section 3]{VanDyke2023}.  We give a detailed account for the case that $\bTheta=\mathbf{(n+1)Vect}$, corresponding to linear bosonic TQFTs, for which we have the following equivalence.
\begin{maintheorem}[Corollary \ref{thm:anomaliesrelative}]\label{mainthm:anomalyrelative}
    Consider anomaly $\alpha\in \Fun(\cX,\rB \nVect)$.  There is an equivalence of categories between $\cX$-families of theories with target $\nVect$ and anomaly $\alpha$, and the category of defects between $\underline{\nVect}$ and  $\alpha$, where $\underline{\nVect}$ is the constant $\cX$-family with value $\nVect\in \mathbf{(n+1)\Vect}$. 
\end{maintheorem}
Our approach to anomalies has the additional advantage that it naturally extends to TQFTs with anomalous categorical symmetries, by replacing the $\infty$-groupoid $\cX$ with an $(\infty,n)$-category $\cC$, and replacing $\Aut$ with $\End$.  We also expect, but have not confirmed, that our Theorems concerning anomalies for crystalline topological phases can be connected with the lattice anomalies in \cite{Shirley:2025yji,Tu:2025bqf}, by using the target given by quantum cellular automata.

  \subsection{Outline}
The contents of this paper is presented as follows: in \S\ref{section:prelim} we provide mathematical definitions for the physical objects appearing in \ref{thm:GCEP} and corollary \ref{cor:CEP}.  In \S\ref{section:anomaliesSPTs} we discuss anomalies for families of TQFTs, compare anomlies for $\rB G$-families with anomalies for $G$-symmetry, and describe anomalous families of $n$-dimensional TQFTs as $n$-dimensional TQFTs relative to an $(n+1)$-dimensional theory. 

\section{Families and $G$-invariant families of TQFTs}\label{section:prelim}
In the following section, we introduce the prerequisites required to formulate and prove Theorem \ref{thm:GCEP}.  It will be necessary to work with the category of $n$-dimensional TQFTs.  We first introduce some notation.  We will refer to \cite{LurieTFT} Definition 2.3.13 and 2.3.16 for further details.  
\begin{defn}\label{def:withduals}
    Let $\bTheta$ be a symmetric monoidal $(\infty,n)$-category.  We say that $\bTheta$ \emph{has duals} if it has duals for objects and, for all $1\leq i<n$, all $i$-morphisms admit adjoints.
\end{defn}

\begin{rem}\label{rem:cobordismreduction}
Take $\mathbf{\Theta}'$ to be a symmetric monoidal $(\infty,n)$-category with duals.  A (fully-extended) $n$-dimensional TQFT with target $\mathbf{\Theta}'$ is a symmetric monoidal $(\infty,n)$-functor $\cob_n\to \mathbf{\Theta}'$ from the (framed) $n$-dimensional cobordism $(\infty,n)$-category $\cob_n$, and codimension-$i$ defects are monoidal lax $i$-transfors.  By the cobordism hypothesis with singularities \cite[Section 4.3]{LurieTFT} and \cite{JohnsonFreydScheimbauer}, one gets a canonical equivalence between the $(\infty,n)$-category of `TQFTs and defects', and the core $\bTheta$ of $\bTheta'$.  
This equivalence allows us to consider $\mathbf{\Theta}$ to be `a category of $n$-dimensional TQFTs'.  Using this equivalence, we can manipulate TQFTs by manipulating the target category $\bTheta$, without reference to the cobordism category.  We will make use of this perspective for the remainder of this work. 
\end{rem}

\begin{defn}\label{def:Xfamily}
   Let ${\mathbf{\Theta}}$ be a category of $n$-dimensional TQFTs. An \emph{$\cX$-family of TQFTs valued in $\bTheta$}, or simply an \emph{$(\cX,\bTheta)$-theory}, is a functor $\mathrm{Th}:\cX \rightarrow \mathbf{\Theta}$.
\end{defn}

\begin{example}
Suppose that $\cX$ is a space equipped with $G$-action.  Theorem \ref{thm:GCEP} provides an abundance of $\cX/\!/G$-families of TQFTs. 
\end{example}

\begin{defn}
    Let $G$ be a group.  Let $\rB G$ denote its \emph{delooping}: the groupoid with a single object, denoted by $\pt\in \rB G$, and morphisms given by $\Hom_{\rB G}(\pt,\pt)\simeq G$.  If $G$ is abelian, this construction can be iterated to give an $i$-groupoid $\rB^iG$. 
\end{defn}

Given a category $\cC$ and an object $c\in \cC$ equipped with $G$-action, one can equivalently view this data as a functor $\rB G\to \cC$, taking the unique object $\pt\in \rB G$ to $c\in \cC$, and taking each $g\in \Hom_{\rB G}(\pt,\pt)$ to the endomorphism $g\cdot-:c\to c$.  

\begin{example}
Suppose that $c\in \cC$.  Denote by $\underline{c}:\rB G\to \cC$ the \emph{constant functor}, which takes all objects in $\rB G$ to $\cC$, and all morphisms to $\id_{c}$.  This is the same as equipping $c$ with the trivial action of $G$. 
\end{example}

\begin{example}
Let $\bTheta$ be a symmetric monoidal $(\infty,n)$-category with duals, which we equivalently view as a category whose objects are $n$-dimensional TQFTs, and whose $i$-morphisms are codimension $i$-defects as in remark \ref{rem:cobordismreduction}. 

Suppose that $G$ is a group.  According to the previous discussion, a $\rB G$-family of theories valued in $\bTheta$, i.e. a functor $\Th:\rB G\to \bTheta$, picks out a theory $\Th(\pt)\in \bTheta$, together with a family of codimension-1 defects labeled by $G$, with fusion rules determined by group multiplication.  

More generally, if $G$ is abelian, a $\rB^{i+1} G$-family of theories valued in $\bTheta$ picks out a theory $\Th(\pt)\in \bTheta$, together with a family of codimension-${i+1}$ defects labeled by $G$, with fusion rules determined by group multiplication.  This is the data of an \emph{$i$-form global symmetry}, in the sense of \cite{GeneralizedGlobal}.  We summarize with the following slogan:
\begin{slogan}
    An $i$-form $G$-symmetry is equivalent to a $\rB^{i+1} G$-family of theories.  
\end{slogan}
We will justify this perspective further in section \ref{section:anomaliesSPTs} by demonstrating that anomalies of $\rB^{i+1} G$-families recover the expected anomalies of $i$-form symmetries. 

In particular, a $\cX$-family of theories should be viewed as a $(-1)$-form $\cX$-symmetry.  In the existing literature, $(-1)$-form symmetries have been interpreted in terms of decomposition phenomena \cite{Sharpe:2022ene,Cordova:2019jnf}, or as space-filling operators \cite{Yu:2020twi}.
\end{example}

\begin{example}
According to the \emph{homotopy hypothesis}, a space equipped with $G$-action is equivalent to a functor $\rB G\to \Cat_{(\infty,0)}$ from $\rB G$ into the category of \emph{$\infty$-groupoids}. 
\end{example}

\begin{example}
   An $(\infty,n)$-category equipped with $G$-action is equivalent to a functor $\rB G\to \Cat_{(\infty,n)}$. 
\end{example}
 
There is a canonical functor $\Cat_{(\infty,n-1)}\hookrightarrow \Cat_{(\infty,n)}$, obtained by viewing an $(\infty,n-1)$-category as an $(\infty,n)$-category with only invertible $n$-morphisms.  A space with $G$-action $\cX$ determines an $(\infty,n)$-category with $G$-action via the composition 
\begin{equation}
\rB G\xrightarrow{\cX} \Cat_{(\infty,0)}\hookrightarrow\Cat_{(\infty,1)}\hookrightarrow\hdots \Cat_{(\infty,n-1)}\hookrightarrow \Cat_{(\infty,n)}. 
\end{equation}

\begin{defn}
    Let $\cX\in \Fun(\rB G, \Cat_{(\infty,n)})$ be an $(\infty,n)$-category equipped with $G$-action.  The \emph{homotopy quotient} of this $G$-category is defined to be the homotopy colimit $\cX/\!/ G:=\mathrm{colim}_{\rB G} \cX$. 
\end{defn}
Given two categories with $G$-action, viewed as functors $\rB G\to \Cat_{(\infty,n)}$, one can check that a $G$-equivariant functor between $G$-categories is equivalent to a natural transformation between the associated functors.  The following is a rephrasing of \cite[Definition 6]{TE1} in this language. 
\begin{defn}\label{def:GinvntXfamily}
    Let $\cX$ be a $G$-space, and $\Theta$ be a symmetric monoidal $(\infty,n)$-category with duals.  A \emph{G-invariant $\mathcal{X}$-family of TQFTs valued in $\bTheta$} is equivalent to a natural transformation from $\cX$ to $\underline{\bTheta}$.  That is, an object of $\Hom_{\Fun(\rB G,\Cat_{(\infty,n)})}(\cX,\underline{\bTheta})$. 
\end{defn} 

\section{Anomalies of $\infty$-groupoid symmetry}\label{section:anomaliesSPTs}
In this section we give a definition of an anomaly for arbitary $\infty$-groupoid symmetry, which in particular applies to anomalies for $\cX/\!/G$-families as constructed by Theorem \ref{thm:GCEP}.  We begin by showing that (nonanomalous) families of theories (Definition \ref{def:Xfamily}) can be equivalently encoded as follows.
\begin{lem}\label{lem:trivialbundle}
    The data of a $\cX$-family of TQFTs valued in $\mathbf{\Theta}$ is equivalent to a section of the trivial fibration $\pi_\cX:\cX \times \mathbf{\Theta}\to \cX$. 
\end{lem}
\begin{proof}
  Given a functor $\mathrm{Th}:\cX\to \mathbf{\Theta}$, one obtains a section of the trivial fibration from the universal property of the product   
\begin{equation}
    \begin{tikzcd}
        &\cX \ar[d,dotted]\ar[dl,swap,"\mathrm{id}_\mathcal{X
        }"]\ar[dr,"\mathrm{Th}"]&\\
        \cX&\cX\times \mathbf{\Theta}\ar[l,"\pi_\cX"]\ar[r,swap,"\pi_\mathbf{\Theta}"]&\mathbf{\Theta}
    \end{tikzcd}
\end{equation}
Conversely, given a section $\gamma$, one obtains a functor $\pi_\mathbf{\Theta}\circ \gamma:\cX\to \mathbf{\Theta}$. \end{proof}

This suggests a natural generalization:  We can replace the trivial fibration $\cX\times \mathbf{\Theta}\to \cX$ with another fibration $\widetilde{\mathbf{\Theta}}\to \cX$ whose fibers over any $x\in \cX$ agree with $\mathbf{\Theta}$. 
\begin{defn}\label{def:anomaly}
An \emph{anomaly for $\cX$-families of theories valued in $\mathbf{\Theta}$} is a fibration $\widetilde{\mathbf{\Theta}} \to \cX$ with fiber $\bTheta$. An \emph{anomalous $\cX$-family of theories with anomaly $\widetilde{\bTheta}$} is a section of $\widetilde{\mathbf{\Theta}}$. 
\end{defn}
We have not yet defined what we mean by a fibration of categories.  We proceed with the relevant background on definitions and tools required to prove Theorem \ref{mainthm:classifyanom} and Theorem \ref{mainthm:anomalyrelative}. 
\subsection{Fibrations of Categories and Straightening/Unstraightening Equivalence}\label{subsection:Unstraightening}
For the purpose of pedagogy, we begin by phrasing the correspondence between covering spaces and fibers in language appropriate for generalization.  

Let $X$ be a topological space, and let $p:W\to X$ be a covering space.  Given a path $\gamma:[0,1]\to X$ with $\gamma(0)=x$ and $\gamma(1)=x'$, and $w\in \pi_1(W)$ with $p(w)=x$, there is a unique lift of $\gamma$ starting at $w\in W$.

We now rephrase this in terms of fundamental groupoids.  The covering space $p:W\to X$ induces a map of fundamental groupoids $p:\pi_1(W)\to \pi_1(X)$.  Given a morphism $\gamma:x\to x'$ in $\pi_1(X)$ and $w\in \pi_1(W)$ with $p(w)=x$, there is a unique morphism $\tilde{\gamma}$ in $\pi_1(W)$, satisfying $p(\tilde{\gamma})=\gamma$, whose source is $w$. 

Then the data of a covering space $p:W\to X$ is equivalent to a map of 1-groupoids $\pi_1(W)\to \pi_1(X)$ with this `lifting property'.  The category of `groupoids $\cZ$ equipped with map $\cZ\to \pi_1(X)$ that fulfills this lifting property', forms a full subcategory of the slice category $\fib_{/\pi_1(\cX)}\subset \mathbf{Gpd}_{/\pi_1(\cX)}$. 

The category $\fib/_{\pi_1(X)}$ is equivalent to another category.  Given a covering space $p:W\to X$, one obtains a functor $\pi_1(X)\to \Set$ which takes a point $x\in \pi_1(X)$ to the fiber $p^{-1}(x)$, and a path to the corresponding map of fibers.  

The covering fundamental groupoid $\pi_1(W)$ can be reconstructed from this `functor of fibers'. One can show that this correspondence defines an equivalence of categories 
\begin{equation}
\Fun(\pi_1(X),\Set)\simeq \fib_{/\pi_1(X)}\,.
\end{equation}

We will need to generalize this correspondence in a few ways.  First, note that a groupoid is a $(1,0)$-category, and a set is a $(0,0)$-category.  Then, the above is really an equivalence
\begin{equation}
\Fun(\pi_1(X),\Cat_{(0,0)})\simeq \fib_{/\pi_1(X)}\subset {\Cat_{(1,0)}}_{/\pi_1(X)}.
\end{equation}

The correspondence still holds if we replace the fundamental groupoid $\pi_1(X)$ with the fundamental $\infty$-groupoid $\Pi(X)$, which remembers not only paths, but all higher homotopies, and similarly replaced our sets with $\infty$-groupoids.  In fact, for an appropriate notion of fibration\footnote{When one allows non-invertible paths (i.e. $(\infty,1)$-categories), there are two notions of fibration, called \emph{Cartesian} and \emph{coCartesian}, corresponding to $\Fun(\cC^{op},\Cat_{(\infty,1)})$ and $\Fun(\cC,\Cat_{(\infty,1)})$, respectively.  More generally, if one allows $(\infty,n)$-categories, there are $2^n$ notions of fibration, corresponding to covariance or contravariance along each layer of non-invertible homotopy.}, one could even allow `non-invertible homotopies'.  This amounts to replacing our $(1,0)$ and $(0,0)$-categories with $(\infty,n)$-categories.  For any $\cC\in \Cat_{(\infty,n)}$, we arrive at an equivalence
\begin{equation}
\begin{tikzcd}
\Fun(\cC,\Cat_{(\infty,n)})\ar[r,shift left=1,"\un_\cC"]&\fib_{/\cC}\ar[l,shift left=1, "\st"]\subset {\Cat_{(\infty,n)}}_{/\cC}\,.  
\end{tikzcd}
\end{equation}
This equivalence is called the \emph{Grothendieck construction} or \emph{straightening/ unstraightening} correspondence.  

Consider a functor $F \in \Fun(\cC,\Cat_{(\infty,n)})$.  Informally, the category $\un_{\cC}(F)$ has objects given by pairs $(c,x)$ where $c\in \cC$ and $x\in F(c)$.  A morphism $(c,x)\to (d,y)$ is a pair of a morphism $g:c\to d$ and a morphism $F(g)(x)\to y$.

\begin{rem}
    The staightening/unstraightening equivalence plays a key role in constructing coherent functors of $(\infty,1)$-categories, and was developed in \cite{LurieHTT}.  We refer to \cite{MazelGeecoCart} for a detailed summary in the context of $(\infty,1)$-categories, including definitions of (co)Cartesian fibrations.  We require a version of the staightening/unstraightening equivalence for $(\infty,n)$-categories, and we refer to \cite{nuiten2024straightening} for proof of this equivalence.
\end{rem}

\begin{example}
   Consider the identity fibration $\mathrm{id}_\cX:\cX\to \cX$.  The fibers of this fibration are trivial, so the straightening $\st(\id_\cX)$ is given by the constant functor $\underline{\pt}:\cX\to \Cat_{(\infty,n)}$, which takes objects in $\mathcal{X}$ to the trivial category $\pt\in \Cat_{(\infty,n)}$, and all morphisms to $\id_{\pt}$.
   
   Consider another fibration $f:\mathcal{V}\to \cX$.  A section of $f$ is taken to a natural transformation $\underline{\pt}\to \st(f)$.  This is the key construction behind the proof of Theorem \ref{thm:section}. 
\end{example}

\subsection{Classification of anomalies}

\begin{thm}\label{thm:catofanomalies}
     The category of anomalies for $n$-dimensional $(\cX,\mathbf{\Theta})$-theories is equivalent to the full subcategory of functors 
    $$\Fun(\cX,\rB\sAut(\mathbf{\Theta}))\subset \Fun(\cX,\Spaces)\,,$$
    on those $\alpha \in \Fun(\cX,\Spaces)$ such that $\alpha(y)\simeq \mathbf{\Theta}$ for all $x\in \cX$.  
\end{thm}
\begin{proof}
Recall that the straightening equivalence takes a fibration $f:\widetilde{\bTheta}\to \cX$ to the functor $\un_\cX(f):\cX\to \Cat_{(\infty,n)}$ which takes $x\in \cX$ to the fiber of $f$ over $x$.  Note that $\rB \End(\bTheta)\subset \Cat_{(\infty,n)}$ is the full subcategory on the object $\bTheta$.  

Then unstraightening induces an equivalence between the category of $(\cX,\bTheta)$-anomalies and the full subcategory of $\Fun(\cX,\Cat_{(\infty,n)})$ on those objects $F\in \Fun(\cX,\Cat_{(\infty,n)})$ which factor through $\rB \End(\bTheta)\subset \Cat_{(\infty,n)}$. 

Since $\cX$ is a groupoid, every such functor factors through the groupoid core $\rB \End(\bTheta)^\times \simeq \rB \sAut(\bTheta)$.  Then unstraightening induces an equivalence between the category of $(\cX,\bTheta)$-anomalies and 
\begin{align*}
    \Fun(\cX, \rB &\sAut(\bTheta))\simeq \Fun(\cX, \rB \End(\bTheta))\subset \Fun(\cX,\Cat_{(\infty,n)}). 
\end{align*}
\end{proof}

\begin{example}
Take $\cX = \rB G$ and $\mathbf{\Theta} = \mathbf{Vect}$, the category of finite dimensional vector spaces.  Notice that $\sAut(\Vect)\simeq \rB \mathbb{C}^\times$.  By Theorem \ref{thm:catofanomalies}, anomalies for $\rB G$-families of 1-dimensional TQFTs valued in $\Vect$ are classified by 
\begin{equation}
\mathrm{H}^2(\rB G;\mathbb{C}^\times)\simeq
\pi_0\Fun(\rB G,\rB^2\mathbb{C}^\times)\simeq 
\pi_0\Fun(\rB G, \rB \sAut(\Vect)). 
\end{equation}
This agrees with the typical classification of anomalies for 0-form $G$-symmetries in $1$-dimensional theories.
\end{example}

\begin{example}
   For $G$ a (sufficiently abelain) group, take $\cX=\rB^{i+1} G$ and $\bTheta=\nVect$.  There is a functor $\rB^{n}\mathbb{C}^\times\to \sAut(\nVect)$, taking $\mathbb{C}^\times$ to the twists by invertible $n$-morphisms of the unit $n$-vector space.  This gives a postcomposition functor
   \begin{equation}
       \mathrm{H}^{n+1}(\rB^{i+1} G,\mathbb{C}^\times)\simeq \pi_0\Fun(\rB^{i+1} G,\rB^{n+1}\mathbb{C}^\times)\to \pi_0\Fun(\rB^{i+1} G,\rB\sAut(\nVect)), 
   \end{equation}
   which recovers the anomalies of $i$-form $G$-symmetries in $n$-dimensional QFTs.  Notice that this is \textbf{not} an isomorphism for all $n$, since one may also have anomalies corresponding to `Witt class' type objects in higher vector spaces.  
\end{example}

\begin{rem}
In the discussion of Lemma \ref{lem:trivialbundle} and Theorem \ref{thm:catofanomalies}, one could instead replace the $\infty$-groupoid $\cX$ with an $(\infty,n)$-category $\cC$ to obtain an equivalence between $(\cC,\bTheta)$-anomalies and the category \begin{equation*}
    \Fun(\cC,\rB \End(\bTheta))\subset \Fun(\cC,\Cat_{(\infty,n)}). 
\end{equation*}
This allows one to describe anomalies of non-invertible and categorical symmetries.  For example, given fusion category $\cD$, anomalies of $0$-form $\cD$-symmetry are described by $\Fun(\rB \cD,\rB \End(\bTheta))$. In the case where we take the target to be $\nVect$ we recover the definition in \cite{Thorngren:2019iar}, where categorical anomalies are obstructions to fiber functors. The fiber functors perspective can be improved to give a better quantification of categorical anomalies, which was done in \cite{Antinucci:2025fjp} using exact sequences of tensor categories.
\end{rem}

\begin{thm}\label{thm:section}
    There is an equivalence of categories between $(\cX,\mathbf{\Theta})$-theories with anomaly $\alpha \in \Fun(\cX,\Spaces)$, and natural transformations $\underline{\pt}\to \alpha$, where $\underline{\pt}$ is the constant functor $\cX\to \Spaces$ valued at $\pt\in \Spaces$. 
\end{thm}
\begin{proof}
Fix an anomaly $\alpha$, which unstraightens to give a fibration $\un(\alpha) \rightarrow  \cX$. We want to show that a section for this fibration corresponds to a natural transformation.
By the (un)straightening correspondence, we have an equivalence 
   \begin{align}\label{eq:nattrans}
       \Hom_{\fib_{/\cX}}(\cX, \un(\alpha))
       \simeq \Hom_{\Fun(\cX, \Spaces)}(\st(\cX),\alpha)
      \simeq \Hom_{\Fun(\cX, \Spaces)}(\underline{\pt},\alpha).
   \end{align}
   Here, $\Hom_{\fib_{/\cX}}(\cX, \un(\alpha))$ corresponds to sections of $\un(\alpha)\rightarrow \cX$, and $\Hom_{\Fun(\cX, \Spaces)}(\underline{\pt},\alpha)$ corresponds to natural transformations $\underline{\pt}\to \alpha$. 
\end{proof}

\subsection{From anomalies to relative theories}\label{subsection:anomliestorelative}
Using the fibrational definition of anomaly given in Definition \ref{def:anomaly}, we now give an interpretation of \emph{anomaly inflow} by explaining how these anomalous $n$-dimensional theories can be interpreted as a theory relative to an $(n+1)$-dimensional theory.  We must pass to a higher categorical level and regard the target $\bTheta$ as an object in a $(n+1)$-category.

Let $\mathbf{\Sigma}$ be a symmetric monoidal $(n+1)$-category with duals, which we view as an `ambient category'. Suppose that $\mathbf{\Theta} \in \mathbf{Alg}_{E_\infty}(\mathbf{\Sigma})$ is a commutative algebra object. 
 
\begin{defn}

Denote by $\rB \End(\mathbf{\Theta})\subset \mathbf{\Sigma}$ the full subcategory on the object $\mathbf{\Theta}$.  An \emph{internal $(\cX,\bTheta)$-anomaly} is an object 
\begin{equation}
    \alpha \in \Fun(\cX,\rB \End(\bTheta))\subset \Fun(\cX, \mathbf{\Sigma}). 
\end{equation}
\end{defn}

\begin{rem}\label{rem:employ}
By cobordism hypothesis with singularities, a morphism in $\Fun(\cX, \bSigma)$ (i.e. a natural transformation) determines a defect between $\cX$-families of theories valued in $\bSigma$, or equivalently, an $n$-dimensional relative theory. 
\end{rem}

We now turn a $(\cX,\bTheta)$-internal anomaly into an anomaly valued in a different target. 
Consider the functor
\begin{equation}
h_{\bTheta}=\Hom_{\bSigma}(\bTheta,-):\bSigma\to \Cat_{(\infty,n)}.
\end{equation}
The image of $\bTheta$ under this functor is $\End_{\bSigma}(\bTheta)$, which inherits the structure of a symmetric monoidal $(\infty,n)$-category with duals, i.e. is a good target for $n$-dimensional TQFTs. Then $h_\bTheta$ restricts to a functor $h_\bTheta:\rB \End(\bTheta)\to \rB \End(\End_{\bSigma}(\bTheta))$, which induces a postcomposition map:
\begin{equation}
\widehat{{h_\bTheta}}:\Fun(\cX, \rB\End(\bTheta))\to \Fun(\cX, \rB \End(\End_{\bSigma}(\bTheta))). 
\end{equation}
In particular, any internal $(\cX,\bTheta)$-anomaly $\alpha$ determines an $(\cX,\End_{\bSigma}(\bTheta))$-anomaly given by $\widehat{{h_\bTheta}}(\alpha)$. 

Suppose now that $h_\bTheta$ has a left adjoint on some full subcategory $\Cat^\pt_{(\infty,n)}\subset\Cat_{(\infty,n)}$ which contains the point and the essential image of $h_\Theta$:
\begin{equation}\label{eq:sigmaadj}
\begin{tikzcd}
\bSigma \,\,\ar[r,shift left=2, "{{h_\bTheta}}"{name=A}]& {\Cat^\pt_{(\infty,n)}}\ar[l,shift left=2,"{{\rint}}"{name=B}].\ar[phantom, from=A,to=B,"\dashv"rotate=90] 
\end{tikzcd}
\end{equation}
We will denote such an adjoint by $\mathrm{int}$, for `internalization'.

\begin{thm}\label{thm:relativegeneral}
Consider a symmetric monoidal $(n+1)$-category $\bSigma$, an algebra $\bTheta\in \mathbf{Alg}_{E_\infty}(\bSigma)$, and internal $(\cX,\bTheta)$-anomaly $\alpha$.  Suppose that $h_{\bTheta}$ has a left adjoint on a subcategory $\Cat^{\pt}_{(\infty,n)}$ as in \ref{eq:sigmaadj}. 

There is an equivalence of categories between the category of $(\cX,\bSigma)$-theories relative to $\alpha$, and the category of $(\cX,\End_\bSigma(\bTheta))$-theories with anomaly $\widehat{h}_{\bTheta}(\alpha)$. 
\end{thm}

\begin{proof}
Since $\Fun(\cX,-)$ is a covariant $(\infty,n+1)$-functor, the adjunction in Equation \ref{eq:sigmaadj}
gives an adjunction of the form:
\begin{equation}
\begin{tikzcd}
\Fun(\cX,
\bSigma)\ar[r,shift left=2, "{\widehat{h_\bTheta}}"{name=A}]&\Fun(\cX, {\Cat^\pt_{(\infty,n)}})\ar[l,shift left=2,"{\widehat{\rint}}"{name=B}].\ar[phantom, from=A,to=B,"\dashv"rotate=90] 
\end{tikzcd}
\end{equation}
 This induces an equivalence 
 \begin{equation}
     \Hom_{\Fun(\cX,\Cat^\pt_{(\infty,n)})}(\underline{\pt},\widehat{h}_\Theta(\alpha))
     \simeq \Hom_{\Fun(\cX,\bSigma)}(\,\widehat{\rint}(\underline{\pt}),\alpha) 
     \simeq \Hom_{\Fun(\cX,\bSigma)}(\underline{{\rint}(\pt)},\alpha)\,,
 \end{equation}
where $\underline{\rint(\pt)}$ is the constant functor valued at $\rint(\pt)$.  By employing the cobordism hypothesis as in Remark \ref{rem:employ}, the sequence of equivalences concludes the proof.
\end{proof}

Informally, this expresses the idea that a theory with anomaly can be understood as a boundary condition between the trivial (vacuum) theory and a nontrivial bulk theory in one higher dimension.

\begin{example}
As a special case of Theorem \ref{thm:relativegeneral}, take $\bSigma = \nnVect$, and $\bTheta = \nVect$.
Recall that there is a forgetful-free adjunction 
\begin{equation}\label{eq:adjunction}
    \begin{tikzcd}
        \RMod_{\nVect}(\Cat_{(\infty,n)})\ar[r,shift left=2,"{\oblv}"{name=A}]& \Cat_{(\infty,n)}  \,\ar[l,shift left=2,"\mathrm{lin}"{name=B}],\ar[phantom, from=A,to=B,"\dashv"rotate=90]
    \end{tikzcd}
\end{equation}
where $\oblv\simeq \Hom_{\RMod_{\nVect}}(\nVect,-)$.  $\lin(\pt)\simeq \nVect$ lands in the full subcategory $\textbf{(n+1)}\Vect \subset \RMod_{\nVect}(\Cat_{(\infty,n)})$. Viewing $\alpha$ as an object in $\Fun(\cX,\mathbf{(n+1)Vect})$, we use the adjunction in Equation \ref{eq:adjunction} to get the following equivalence between categories of natural transformations 
\begin{multline}
    \Hom_{\Fun(\cX,\Cat^{\pt}_{(\infty,n)})}(\underline{\pt},\oblv(\alpha))
     \simeq \Hom_{\Fun(\cX,\mathbf{(n+1)Vect})}(\lin(\underline{\pt}),\alpha)\\
     \simeq \Hom_{\Fun(\cX,\mathbf{(n+1)Vect})}(\underline{\nVect},\alpha)
\end{multline}
Combined with Theorem \ref{thm:section} we see that anomalous theories valued in $\nVect$ are natural transformations from $\underline{\nVect}$ to $\alpha$, where $\underline{\nVect}$ is the constant functor valued at $\nVect$.

\end{example}

\begin{cor}\label{thm:anomaliesrelative}
Fix internal $(\cX,\nVect)$-anomaly $\alpha\in \Fun(\cX,\rB\ \End(\nVect))\simeq \Fun(\cX ,\rB\ \nVect)$.  There is an equivalence of categories between theories relative to $\alpha$, and $(\cX,\nVect)$-theories with anomaly $\alpha$. 
\end{cor}

\noindent
\textbf{Acknowledgments:} \emph{It is a pleasure to thank Arun Debray  and Ryan Thorngren for informative discussions.  DS is supported by VILLUM FONDEN, VILLUM Young Investigator grant 42125. MY is supported by the EPSRC Open Fellowship EP/X01276X/1, and would like to especially thank Abhinav Prem for enlightening discussions that led to initial interest in the Crystalline Equivalence Principle. No authors have competing interests to declare that are relevant to the content of this article. There is no data available for this article to declare.}

\bibliographystyle{alpha}
\bibliography{ref}
\end{document}